\documentclass[a4paper,12pt]{article}
\usepackage[left=30mm,top=30mm,right=30mm,bottom=30mm]{geometry}
\usepackage{etoolbox} 
\usepackage{booktabs}
\usepackage[table,xcdraw]{xcolor}
\usepackage[usestackEOL]{stackengine}
\usepackage[round]{natbib}
\usepackage[T1]{fontenc}
\usepackage[utf8]{inputenc}
\usepackage{bm}
\usepackage{graphicx}
\usepackage{subcaption}
\usepackage{amsmath}
\usepackage{amsfonts}
\usepackage{mathtools}
\usepackage{xcolor}
\usepackage{float}
\usepackage{hyperref}
\usepackage[capitalise]{cleveref}
\usepackage{enumitem,kantlipsum}
\usepackage{amssymb}
\usepackage{amsbsy}
\usepackage{amsthm}         

\usepackage[ruled,vlined]{algorithm2e}
\usepackage{listings}

\newcommand{\E}{\mathcal{E}}
\newcommand{\R}{\mathbb{R}}

\newtheorem{theorem}{Theorem}
\newtheorem{lemma}{Lemma}
\newtheorem{proposition}{Proposition}
\newtheorem{definition}{Definition}

\newtheorem{remark}{Remark}
\newtheorem{assumption}{Assumption}
\newtheorem*{assumption*}{Assumption}

\newcommand\independent{\protect\mathpalette{\protect\independenT}{\perp}}
\def\independenT#1#2{\mathrel{\rlap{$#1#2$}\mkern2mu{#1#2}}}

\hypersetup{
    colorlinks,
    linkcolor={black},
    citecolor={blue!50!black},
    urlcolor={blue!80!black}
}

\title{Detecting causal covariates for extreme dependence structures} 
\author{Juraj Bodík$^{1 \footnote{ Email: Juraj.Bodik@unil.ch}}$, Linda Mhalla$^{1,2}$, Valérie Chavez-Demoulin$^1$}
\date{%
    $^1$ {\small Faculté des Hautes Études Commerciales, Université de Lausanne, Switzerland} \\%
    $^2${\small Institute of Mathematics, École polytechnique fédérale de Lausanne, Switzerland}\\[2ex]%
    }

\begin{document}

\pagenumbering{gobble}

\maketitle
\begin{abstract}
     Determining the causes of extreme events is a fundamental question in many scientific fields. An important aspect when modelling multivariate extremes is the tail dependence. In application, the extreme dependence structure may significantly depend on covariates. As for the general case of modelling including covariates, only some of the covariates are causal. In this paper, we propose a methodology to discover the causal covariates explaining the tail dependence structure between two variables. The proposed methodology for discovering causal variables is based on comparing observations from different environments or perturbations. It is a desired methodology for predicting extremal behaviour in a new, unobserved environment. The methodology is applied to a dataset of $\text{NO}_2$ concentration in the UK. Extreme $\text{NO}_2$ levels can cause severe health problems, and understanding the behaviour of concurrent severe levels is an important question. We focus on revealing causal predictors for the dependence between extreme $\text{NO}_2$ observations at different sites. 

    \end{abstract}
    
\newpage

\pagenumbering{arabic}

\section*{Introduction}

Modelling extreme events is a crucial part of many real-world problems and the focus is often on the extremal dependence structure between several variables \citep{Coles}. In a univariate case, under very general assumptions, a normalized maxima of random variables converges to a Generalized Extreme Value (GEV) distribution \citep{GEV}. In the multivariate case, however, the limiting dependence structures for extreme values do not follow a simple parametric form. Therefore, alternative joint tail representations should be used, and extensive literature describes different approaches \citep{Resnick, beirlant2006statistics}.  

More recent literature describes the modelling and estimation of covariate influence on the joint tail structure. \cite{Linda3} propose a parametric framework for integrating covariate information through the angular density model. Non-parametric modelling of the joint tail structure using GAM modelling \citep{Wood} has been developed in \cite{Linda}. In this paper, our main emphasis is on the causal properties of the extremal  tail coefficient between two random variables $Z_1, Z_2$ \citep{Joe}.

 A real-world problem of interest consists of $\text{NO}_2$ measurements in the UK. Features of $\text{NO}_2$ concentration levels or other air pollutants are of interest in environmental science \citep{Shi_2014} as extreme values of $\text{NO}_2$ can be very dangerous and have been linked to an increased risk of lung cancer \citep{NO2cancer}. 
The dataset taken from \cite{openair} comprises 241 measurement stations in the UK with hourly records of $\text{NO}_2$ [µg/m3]. 
We use only the stations with at least 15 years of measurements (21 stations in total). 
$\text{NO}_2$ is partially produced by burning fossil fuel and motor vehicle exhaust. Hence, we distinguish traffic stations where the $\text{NO}_2$ level is predominantly determined by nearby
traffic (12 stations) from background stations  (9 stations).
Figure \ref{FigureTraf} shows $\text{NO}_2$ concentration (in $\mu$g/m$^3$) for one traffic station and one background station over time. We observe that the magnitude is usually larger in traffic stations than in background stations. Although a yearly seasonality seems clearer for the background station, both series highlight a long term decreasing trend. The point of our interest is the extremal dependence between a pair of stations explained by some covariates. Some of these covariates may be causal, others not. We are interested in answering the following question: What is the \textit{cause} of the extremal dependence between  $\text{NO}_2$ concentrations at two stations?

\begin{figure}[h]
\centering
\includegraphics[scale=0.6]{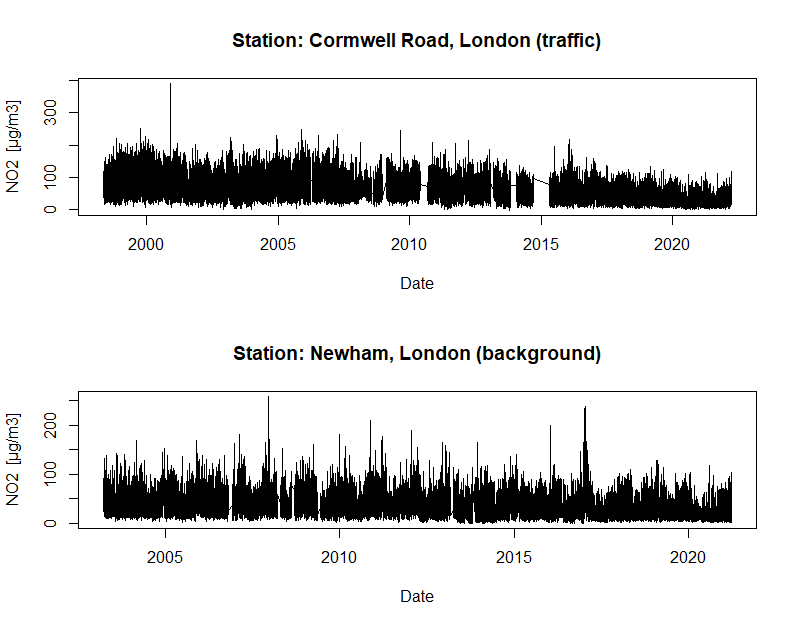}
\caption{An example of traffic and background stations in London. }
\label{FigureTraf}
\end{figure}

Determining a system's causal structure is an important problem in many scientific fields. There is a growing literature on causal inference under various types of assumptions and different frameworks, with applications ranging from public health to biology and economics \citep{Biology_causal_inference, Economics_causal_inference}. Often, the goal is to find the causes of some response variable $Y$ among a given set of covariates $\textbf{X}$ or to quantify the causal relationships between a set of variables. In recent years, a causal discovery based on conditional invariance has become an important growing research area. Starting with a groundwork paper \citep{Peters}, connections between distributional invariance and causality started to emerge. The reviewing paper \cite{Buhlmann} summarises the groundwork and highlights the importance of causal prediction. The methodology based on invariant causal prediction (ICP) relies on observing a system under different environments \citep{Pearl_book}. Then, one can observe some invariant properties for causal variables through all environments, but not for non-causal ones. 

Here, we are interested in unveiling causal links explaining extremal dependence between random variables. Precisely, we want to assess whether an intervention on a covariate affects the tail dependence of the random variables. Classical causal techniques that usually give a good description of data tend to fail in the upper tail due to the scarcity of observations \citep{EngelkeSparsity}. A combination of causality and extreme value theory (EVT) is a developing research area. \cite{Linda_causality} use a Kolmogorov complexity to detect causal direction using observational data.  \cite{10.1214/20-AOAS1355} studied probabilities of necessary and sufficient causation as defined in the counterfactual theory using the multivariate generalized Pareto distributions. \cite{Gnecco} define a causal tail coefficient revealing the causal structure between two random variables by capturing asymmetries in the extremal dependence of the two random variables under the context of the linear structural causal model. Their methodology holds also in the presence of latent common causes with the same tail heaviness as the observed variables. \cite{Pasche} proposes a methodology relying on the causal tail coefficient that can deal with latent common causes that have different tail indices.  \cite{Deuber} describe a framework for estimating extremal quantile treatment effects. However, to our knowledge, no theory combines EVT and causal inference using the ICP framework and inferring causes of the extremal dependence structure.

A well-established approach in causal inference relies on the structural causal model (SCM) of  \cite{Pearl_book}. The ICP framework exploits invariance properties in SCM for inference; see \cite{Stara_causality, Dawid,ML_causal} for similar notions.  The main emphasis in the ICP setting has been on linear regression models assuming the additive form $Y=\textbf{X}_{S^\star}\beta + \varepsilon$, where $Y$ is the response variable of interest, $\textbf{X}_{S^\star}$ is the set of direct causal covariates among the set of covariates $\textbf{X}$ and $S^\star\subseteq\{1, \dots, d\}$ the set of indices of direct causes of $Y$. Generalizations of this model have been given by \cite{Christina} with the more flexible form $Y=g(\textbf{X}_{S^\star}) + \varepsilon$ where $g$ is a smooth, non-parametric function. We aim to use the ICP methodology in the context of the joint tail structure.  

The paper is organized as follows: Section 1 introduces EVT in the bivariate context. Section 2 describes the ICP methodology. Section 3 proposes a methodology to unveil  causal covariates among a set of covariates explaining the extremal dependence between two random variables. A simulation study is presented in Section 4 and Section 5 illustrates the proposed methodology through the dataset of $\text{NO}_2$ measurements in the UK.


\section{Bivariate extreme value theory}


Let $(Y_i)_{i\geq 1}$  be a sequence of independent and identically distributed (iid) random variables
with common distribution $F$. We are interested in the behaviour of the maximum
of a sequence of $n$ such random variables, that is $M_n=\max_{i=1, \dots, n} Y_i$. If there exist a non-degenerate distribution $G$ and sequences of constants $\{a_n>0\}$ and $\{b_n\}$ such that $\frac{M_n-b_n}{a_n}\overset{D}{\to}G$ as $n\to\infty$, then $G$ belongs to the Generalized Extreme Value (GEV) distributions family. The GEV distribution has three parameters, namely a location $\mu\in\mathbb{R}$, a scale $\sigma>0$, and a shape $\xi\in\mathbb{R}$. Its distribution function takes the form 
\begin{equation}
\label{GEV}
  \begin{aligned}
  F(x) &=  \exp\left\{ - \left(1+\xi\frac{x-\mu}{\sigma}\right)^{-1/\xi}  \right\} \text{,    }\xi\neq 0,\\ 
F(x) &=  \exp\left\{ - \exp\left(-\frac{x-\mu}{\sigma}\right) \right\} \text{,    }\xi= 0,
   \end{aligned}
\end{equation}
defined on the support $[\mu - \sigma/\xi, \infty), (-\infty, \infty), (-\infty, \mu - \sigma/\xi]$ for cases $\xi<0, \xi=0, \xi>0$ respectively. Cases when $\xi>0, \xi=0$, and $\xi<0$ correspond to the Fréchet, Gumbel, and Weibull distributions, respectively \citep{fisher_tippett_1928, GEV}. 

Suppose now $(\textbf{Y}_i)_{i\geq 1}$ is a sequence of independent copies of a two dimensional random vector $\textbf{Y}=(Y_1,Y_2)^\top$ with marginal distributions $F_j$, $j=1,2$ and joint distribution $F$ and denote by 
$\textbf{M}_n=(M_{n,1}, M_{n,2})^\top$ their component-wise maxima. 
We suppose that there exist sequences 
$\{\textbf{a}_n\}\subset \R^2_+$ 
and $\{\textbf{b}_n\} \subset \R^2$ such
that the normalized vector $(\textbf{M}_n - \textbf{b}_n)/\textbf{a}_n$ converges in distribution to a random vector $\textbf{Z}$ 
with joint distribution $G$ and non-degenerate margins $G_j$, $j = 1,2$.
Then, the limiting distribution $G$ is called the multivariate extreme value distribution
(MEVD) \citep[Chapter 5.3]{Resnick}. The class of MEVDs coincides with the class of max-stable distribution functions with non-degenerate
margins. The margins $G_j$ are max-stable, or equivalently GEV. Without loss of generality, it is convenient to transform the marginal
distributions $G_j$ to unit Fréchet margins with distribution function 
$P(Z_j\leq z)=\exp{(-1/z)}1_{[0,\infty)}(z)$. The transformation is marginally performed using a probability integral
transform with 
\begin{equation}\label{FrechMarg}
    t_j(z)= -1/\log{G_j(z)}, z>0,
\end{equation}
applied to the $j$th GEV margin $G_j$.
When the margins are transformed to unit Fréchet, the dependence structure of the transformed random vector is the same as  the dependence structure of the  vector $\textbf{Z}$. We therefore assume hereafter that $\textbf{Z}$ has unit Fréchet margins. In that case, $G$ is called a simple max-stable distribution and has the representation $G(\textbf{x})=e^{-V(\textbf{x})}$, where $V(\textbf{x})$ is called the exponent function \citep[Chapter 5.4]{Resnick}. It can be shown that $V$ is $-1$ homogeneous, that is, $V(t\textbf{x})=\frac{1}{t}V(\textbf{x}), t>0, \textbf{x}=(x_1,x_2)\in\mathbb{R}^2_+$ \citep{Coles2}. The Pickands dependence function \citep{Pickands1981} $A: \mathcal{S}_2\to [\frac{1}{2}, 1]$ is defined by 
\begin{equation}\label{PickandsPreliminaries} 
 V(\textbf{x}) =\bigg(\frac{1}{x_1} + \frac{1}{x_2}\bigg) A(\boldsymbol{\omega}),
 \end{equation}
where $\mathcal{S}_2$ is the unit simplex and $\boldsymbol{\omega} = (\omega_1, \omega_2)= (\frac{x_1}{\Vert \textbf{x}\Vert},\frac{x_2}{\Vert\textbf{x}\Vert} )\in\mathcal{S}_2$, for any norm $\Vert \cdot \Vert$. The Pickands dependence function is the restriction of the exponent function $V$ to $\mathcal{S}_2$, and it measures the strength of extremal dependence along a direction $\boldsymbol{\omega} \in \mathcal{S}_2$. 
This function must satisfy the boundary restrictions $\max{(\omega_1,\omega_2)}\leq A(\boldsymbol{\omega})\leq 1$  where the lower bound
corresponds to perfect extremal dependence and the upper bound to the independence
case. It also satisfies the
border conditions $A(0,1)=A(1,0) = 1$ and is a continuous and convex function on $S_2$.
The extremal coefficient \citep{Smith}, defined by 
\begin{equation}
    \theta= V(1,1)=2A(1/2,1/2),
    \label{theta}
\end{equation}
is a summary of the strength of extremal dependence between $Z_1$ and $Z_2$. Noting that $P(\textbf{Z}\leq z) = \exp(-1/z)^{\theta}$, 
the extremal coefficient lies between 1 (perfect dependence) and 2 (independence). Moreover,  
\begin{equation}\label{chi}
\theta=2-\lim_{u\to\infty} P(Z_1>u\mid Z_2>u).
\end{equation}
In the rest of the paper, we focus mainly on modelling $\theta$. 

For modelling $G$, it is common to use parametric representations of the exponent function $V$. An example is the logistic model, defined as 
\begin{equation}\label{LogisticModel}
P(Z_1\leq z_1, Z_2\leq z_2) = \exp\{ -\big(z_1^{-1/\alpha} + z_2^{-1/\alpha}\big)^{\alpha} \} , \alpha\in (0,1]. 
\end{equation}
Independence and complete dependence correspond to  $\alpha=1$,  $\alpha\to 0$, respectively \citep{Tawn1988}. It is a symmetric model, where $Z_1, Z_2$ are exchangeable and a simple relation $\theta = 2^\alpha $ holds. Another example is the Hüsler--Reiss model, defined as  
\begin{equation}\label{HR}
P(Z_1\leq z_1, Z_2\leq z_2) = \exp \bigg\{  -\Phi\big(\lambda + \frac{ z_2-z_1}{2\lambda}\big)e^{-z_1} - \Phi\big(\lambda + \frac{z_1- z_2}{2\lambda}\big)e^{- z_2} \bigg\}, 
\end{equation}
where $\Phi$ is a Gaussian distribution function and $\lambda\in(0,\infty)$ is a parameter. Independence and complete dependence correspond to  $\lambda\to \infty$,  $\lambda\to 0$, respectively, and $\theta = 2\Phi(\lambda)$ holds \citep{HuslerReiss}.  Other commonly used parametric models include, for example, Dirichlet models \citep{TAWN1990}.

In many applications, the extremal dependence structure depends on covariates. Consider the bivariate max-stable random vector $\textbf{Z}(\textbf{x}) = (Z_1(\textbf{x}), Z_2(\textbf{x}))^\top$ given a set of $d$ observed covariates $\textbf{x}=(x_1,\ldots, x_d)^\top$ and suppose that  it has unit Fréchet marginal distributions. Both the related Pickands function and extremal coefficient are therefore covariate-dependent and can be written $A(\boldsymbol{\omega}|\textbf{x})$ for $\boldsymbol{\omega}\in S_2$ and $\theta(\textbf{x}) = 2A\left\{(\frac{1}{2},\frac{1}{2})\mid \textbf{x}\right\}$. While our emphasis is on non-stationary extremal dependence structures, one may also consider the margins of $\textbf{Z}(\textbf{x})$ to be covariate-dependent, that is, $F_{j,\textbf{x}}(z) = P(Z_j(\textbf{x})\leq z)$ for $j=1,2$.  In that case, a regression model applied to the marginal GEV distributions’ parameters \citep{Pauli} is required to transform the margins to a unit Fréchet scale.

In order to infer about the covariate-dependent Pickands function, \cite{Linda} introduced, for a fixed direction in the unit simplex, the max-projection
$$
Y(\omega, 1-\omega)\mid \textbf{x} = \max[H^{\lambda_1}\{Z_1(\textbf{x})\}, H^{\lambda_2}\{Z_2(\textbf{x})\} ], \quad \omega\in(0,1),
$$
where $\lambda_1 = \omega (\frac{1}{\omega} + \frac{1}{1-\omega}) , \lambda_2 = (1-\omega) (\frac{1}{\omega} + \frac{1}{1-\omega})$, and $H(x) = e^{-1/x}$. It was shown that the resulting random variable is Beta distributed with parameters $A\left\{(\omega,1-\omega)\mid \textbf{x}\right\}$ and $1$. Hence, for a fixed value of $\omega$, statistical inference for $A\left\{(\omega,1-\omega)\mid \textbf{x}\right\}$ can be performed through likelihood procedures or regression techniques.

\section{Causal modelling} 
\label{sectionCausalInfernece}
First, we introduce a graphical causal notation (e.g., \cite{PCalgorithm}). 
A DAG (directed acyclic graph) $\mathcal{G}=(V,E)$ contains a finite set of vertices (nodes) $V$ and a set of directed edges $E$  between distinct vertices. We assume that there exist no directed cycle or multiple edges. Let $i,j\in V$ be two distinct nodes of $\mathcal{G}$. We say that $i$ is a parent of $j$ if there exist an edge from $i$ to $j$ in $E$, and denote $i\in  pa_{j}(\mathcal{G})$. We usually omit $\mathcal{G}$ in the notation if it is clear from context. Consider a random vector $(X_i)_{i\in V}$ over a probability space $(\Omega, \mathcal{A}, P)$. With a slight abuse of notation, we identify the vertices $j \in V$ with the variables $X_j$.  We denote $\textbf{X}_S = \{X_s: s\in S\}$ for $S\subseteq V$. Typically, we denote $V = \{0, \dots, d\}$, where $X_0$ (usually denoted by $Y$) is a variable of interest and $\textbf{X}=(X_1, \dots, X_d)^\top$ are covariates. We assume that $(Y,\textbf{X})$ follow a structural causal model (SCM)  \citep{Pearl_book}, defined as follows. 

\begin{definition}\label{def1}
The random variables $X_0, \dots, X_d$ (where $X_0$ plays the role of the target variable $Y$) follow a SCM with DAG $\cal{G}$ representing the causal structure \citep{Pearl} if for each $i\leq d$ the variable $X_i$ arises from a structural equation $$
X_i=g_i(\textbf{X}_{pa_i}, \varepsilon_i),
$$
where $\varepsilon_i$ are jointly independent random variables. The functions $g_i$ are the assignments (or link functions) and we say that $X_j$ is a direct cause of $X_i$, if $j\in pa_i(\mathcal{G})$. 
\end{definition}

We are interested in estimating the set of direct causes of $Y$. Doing that only from a  random sample is often impossible, unless we assume some identifiable structure and additional strong assumptions \citep{Lingam,Peters2014}.

\begin{definition}[Intervention]
Consider an SCM as in Definition \ref{def1}. We replace one of the structural equations to obtain a new SCM. Assume that we replace the assignment for $X_i$ by $X_i=\tilde{g}_i(\textbf{X}_{\tilde{pa}_i}, \tilde{\varepsilon}_i)$, where the new $\tilde{\varepsilon_i}$ is independent of all $\varepsilon_j$ and $\tilde{g}_i, \tilde{pa_i}$ are potentially arbitrary. We then call the entailed distribution of the new SCM an intervention distribution and say that $X_i$ has been intervened on. 
\end{definition}

A classical example of an intervention is a hard (atomic) intervention on $X_i$ \citep[Definition 6.8.]{Elements_of_Causal_Inference}, replacing the structural equation in the SCM by $X_i=const$. Note the difference between a simple conditioning $Y\mid X_i=const$ and a hard intervention, where we assume that the joint distribution $P_{(Y^e, \textbf{X}^e)}$ results from a different data-generating process. 

More generally, suppose that we observe $(Y,\textbf{X})$ in different environments \citep[Section 7.1.6]{Elements_of_Causal_Inference}. In each environment $e\in\mathcal{E}$, there is a target variable $Y^e$ and covariates $\textbf{X}^e=(X_1^e, \dots, X_d^e)^\top$. The number of environments $|\mathcal{E}|$ is assumed finite. The set of environments $\mathcal{E}$ is defined in a general way by assuming that for each $e\in\mathcal{E}$, $(Y^e, \textbf{X}^e)$ follows (a different) SCM with distribution $P^e = P_{(Y^e, \textbf{X}^e)}$. A different environment can be generated by (possibly several) interventions \citep{Interventions} on the covariates. However, we do not assume that any information about the interventions is known. 

Note that interventions on $\textbf{X}$ do not change the conditional distribution of $Y\mid \textbf{X}$. We assume that the parents of $Y^e$ remain the same in all environments. Consider the set $S^\star\subseteq \{1, \dots, d\}$ of indices determining the covariates of $ \textbf{X}$ that are true direct causes  of $Y$. Invariance based causal discovery \citep{Peters} relies on the fact that $Y^e\mid \textbf{X}^e_{S^\star}$ is invariant for each $e\in\mathcal{E}$. We can put it more formally by assuming the following:

\begin{assumption}
\label{Assumption1}
There exist  $S^\star\subseteq \{1, \dots, d\}$, a function $g:\R^{|S^\star|+1}\to\R$ and a distribution $F_\varepsilon$, such that for all $e\in\mathcal{E}$ we have 
\begin{equation}
\label{equation1}
Y^e=g(\textbf{X}^e_{S^\star}, \varepsilon^e), \,\, \varepsilon^e \independent  \textbf{X}^e_{S^\star}, \,\, \varepsilon^e\sim F_\varepsilon.
\end{equation}
\end{assumption}

In other words, the structural equation for the target variable can be written as  $
Y^e=g(\textbf{X}^e_{S^\star}, \varepsilon^e)
$ for \textit{all} $e\in\mathcal{E}$. The link function $g$ remains the same in all environments, and noise variables $\varepsilon^e$ are iid. As discussed in \cite{Buhlmann}, Assumption~\ref{Assumption1} guarantees that on one hand $e$ does not change the mechanism between $\textbf{X}^e$ and $Y^e$, and on the other hand that $e$ does not act directly on $Y$ and ideally, $e$ should change the distribution of $\textbf{X}^e$.

In what follows, the formulation ``Assumption \ref{Assumption1} is valid for set $S$”, $S\subseteq\{1, \dots, d\}$ means that there exist $g$ and $F_{\varepsilon}$ such that (\ref{equation1}) holds with set $S$. It is important to observe the following \citep[Section 6.1]{Peters}: 

\begin{lemma}
Assumption 1 is valid for the set S if and only if for all $e,f\in\E$ we have $Y^e\mid (\textbf{X}^e_{S}=\textbf{x})\overset{D}{=}Y^f\mid (\textbf{X}^f_{S}=\textbf{x})$ for all $\textbf{x}$. 
\end{lemma}

The ICP methodology stems from testing for \textit{each} $S\subseteq \{1, \dots, d\}$ the following hypothesis:
$$H_{0, S}(\mathcal{E}): \text{Assumption 1 is valid for the set S}.$$Therefore, 
\begin{equation}
\label{hat(S)}
\hat{S}^\star= \bigcap_{S: H_{0,S}(\mathcal{E})\text{ is not rejected}} S
\end{equation}
is the proposed estimator of the set of direct causes $S^\star$. 

If Assumption 1 holds and we have a proper test for $H_{0,S}(\mathcal{E})$, then $H_{0, S^\star}$ will not be rejected on a given confidence level, and therefore $\hat{S}^\star \subseteq S^\star$. On the other hand, if we want to guarantee $\hat{S}^\star = S^\star$, we need to observe ``enough” environments. It is a big challenge to show some necessary and sufficient conditions for the full identifiability of the causal variables. Some results are known in only special cases \citep{Peters_sparsity}.  It is easy to see that the more environments we observe, the larger $\hat{S}^\star$ is. 

As for testing $H_{0,S}(\mathcal{E})$, we need to assume some form of the link function $g$. Usually, classical regression methods are used to infer about the link function $g$ and  testing $H_{0,S}(\mathcal{E})$ is performed through a testing procedure such as a $k$-sample test. However, if we have a test $H_{0,S}(\mathcal{E})$ having a proper level $\alpha$, the Theorem \ref{Theorem_Level} guarantees that each element of $\hat{S}^\star$ indeed lies in $S^\star$ with probability $1-\alpha$. 

\begin{theorem} \citep{Peters}
\label{Theorem_Level}
Let Assumption \ref{Assumption1} be valid. Consider a test $H_{0,S}(\mathcal{E})$ that has a level $\alpha$ for $S\subseteq \{1, \dots, d\}$. Then, the previously defined estimate $\hat{S}^\star$ fulfills 
$$
P(\hat{S}^\star\subseteq S^\star)\geq 1-\alpha. 
$$
\end{theorem}

Hence, a proper test $H_{0,S}(\mathcal{E})$ that has a level $\alpha$ (at least asymptotically) is needed. One of the goals of the work by \cite{Christina} was to conduct an appropriate $H_{0,S}(\mathcal{E})$ test under different assumptions on the link function. Here, we use the invariant residual distribution test explained below. 

\subsection{Invariant residual distribution test}
\label{Invariant residual distribution test}

We assume that the link function $g$ from Assumption \ref{Assumption1} is additive, that is, the structural equation is of the form  $$Y^e=g(\textbf{X}^e_{S^\star}) + \varepsilon^e, \,\, \varepsilon^e \independent  \textbf{X}^e_{S^\star}, \,\, \varepsilon^e\sim F_\varepsilon$$ 
for all $e\in\mathcal{E}$. The core of the proposed test for $H_{0,S}(\mathcal{E})$ has two main steps:

\begin{enumerate}
\item Pool the data together from all environments and fit a model to predict $Y$ from $\textbf{X}_S$. Stated differently, estimate $g$ relying on data from all environments. 
\item Use a $k$-sample test to assess whether the residuals obtained from different environments have the same
distribution. 
\end{enumerate}

A more detailed description can be found in Appendix II.5 in \cite{Christina}.

\section{Causal modelling of extremal dependence}

In order to frame the modelling of extremes into a structural causal setup, we start by introducing the LogMax-projection.

\begin{lemma}
\label{LLMax-projection}
Assume that $Z=(Z_1, Z_2)^\top$ is a bivariate simple max-stable random vector with a Pickands dependence function $A_{Z_1, Z_2}(\cdot)$ and extremal coefficient $\theta$. Denote $A:=A_{Z_1, Z_2}(\frac{1}{2},\frac{1}{2})$. 
Let 
\begin{align*}
Y':=&\min\bigg(\frac{2}{Z_1}, \frac{2}{Z_2}\bigg),\\
Y:=&  \log\{\max(Z_1, Z_2)\}-\gamma ,
\end{align*}
where $\gamma\approx 0.57$ is the Euler constant. Then, 
\begin{itemize}
\item $Y'$ is exponentially distributed with parameter $A$. In particular, $\mathbb{E}(Y')=\frac{1}{A}$. 
\item $Y\overset{D}{=} \log(\theta) + \varepsilon$, where $\varepsilon\sim Gumbel(-\gamma,1)$\footnote{$Gumbel(\mu, \beta)$ distribution function has the form $P(\varepsilon\leq x) = e^{-e^{-(x-\mu)/\beta}}$. Specifically, $Gumbel(-\gamma, 1$) is a centered Gumbel distribution.}. In particular, $\mathbb{E}(Y)=\log(\theta)$.   
\end{itemize}
\end{lemma}
\begin{proof}
From $P\{\max(Z_1, Z_2)\leq x\}=e^{-\frac{\theta}{x}}$ and $\theta=2A$, we get
$$e^{-A x} = P\{\max(Z_1, Z_2)\leq \frac{2}{x}\}=P\{\min\bigg(\frac{2}{Z_1}, \frac{2}{Z_2}\bigg)\leq x\}.$$ Similar results can be found in \cite{Linda} using the beta distribution. The results for $Y$ follow easily when connecting $Y=-\log\big(\frac{Y'}{2}\big) -\gamma$. We use the fact that the minus logarithm of an exponential distribution is a Gumbel distribution and $\theta=2A$. The mean of the standard Gumbel distribution is $\gamma$, so we obtain centred Gumbel noise by subtracting $\gamma$. 
\end{proof}

We will refer to $Y$ defined in Lemma \ref{LLMax-projection} as a LogMax-projection. This transformation is very useful when the interest is  to make inference on $\theta$. This can be done for instance through classical regression tools. More importantly, it has the desired property that the noise acts additively. Hence, it can be fitted into the causal framework of \cite{Christina}. 

\begin{remark}
Lemma \ref{LLMax-projection} can be used to model the behaviour of the Pickands dependence function evaluated at the specific value $(\frac{1}{2},\frac{1}{2})$. Arguably, this value is the most important in most applications due to its interpretability and connections to the tail dependence measures such as (\ref{chi}). However, we can easily modify Lemma~\ref{LLMax-projection} for any value $\boldsymbol{\omega} = (\omega,1-\omega)$, $\omega \in (0,1)$,  by taking $Y=\log\{\max((1-\omega)Z_1, \omega Z_2)\}-\gamma$. 
\end{remark}

\subsection{Covariate-dependent structures and GAM}
Let us move now to the case where the extremal coefficient $\theta$ depends on a set of covariates $\textbf{X}\in\mathcal{X}\subseteq\mathbb{R}^d$. Consider a bivariate simple max-stable random vector $Z(\textbf{x})=(Z_1(\textbf{x}), Z_2(\textbf{x}))^\top$with a covariate-dependent extremal coefficient. For any $\textbf{x}\in\mathcal{X}$, we denote by $\theta(\textbf{x})$, the extremal coefficient between $Z_1, Z_2$ conditioned on $\textbf{X}=\textbf{x}$. 

From Lemma \ref{LLMax-projection}, we have that the LogMax-projection is defined as $$Y(\textbf{x}):= \log\{\max(Z_1(\textbf{x}), Z_2(\textbf{x}))\}-\gamma ,$$
and satisfies
$$Y(\textbf{x})\overset{D}{=}\log\{\theta(\textbf{x})\} + \varepsilon, \text{ where  }\varepsilon\sim Gumbel(-\gamma, 1).$$ As mentioned above, one can use this model formulation to make inference on $\theta(\textbf{x})$. We propose using the flexible generalized additive model (GAM) \citep{Wood} for the estimation of $\theta(\textbf{x})$. Under the GAM framework,  we make the following assumptions:

\begin{assumption}
Let $\theta(\textbf{x})=h\{\textbf{u}^\top\pmb{\beta} + \sum_{k=1}^K h_k(t_k)\}$, where:

\begin{itemize}
\item $h: \mathbb{R}\to [1,2]$ is a link function,
\item $(u_1, \dots, u_s)$ and $(t_1, \dots, t_K)$ are subsets of $\textbf{x}$, and
\item $h_k: \mathbb{T}_k\to\mathbb{R}$ are functions admitting a finite quadratic penalty representation (\cite{GAM}) supported on a closed set $\mathbb{T}_k\subset \mathbb{R}$ for all $k$. 
\end{itemize}
\end{assumption}
Note that in application, ignoring the fact that the extremal coefficient lies in $[1,2]$ and considering that it is real-valued instead, has no impact on the set of causal covariates discovery. For instance, we do not aim to estimate the extremal coefficient, per se, but rather compare the effect of the covariates under different environments. This does not affect the consistency.    
The estimation of $\theta(\textbf{x})$ is based on maximizing the penalized log likelihood:
\begin{equation}\label{GAMLikelihood}
l(\Gamma, \gamma) = l(\Gamma)  - \sum_{k=1}^K\gamma_k \int_{\mathbb{T}_k} h_k''(t_k)^2dt_k,
\end{equation}
where $\Gamma$ represents the set of unknown parameters ($\beta$ and the linear basis coefficients of each of the smooth functions $h_k$) and $l(\Gamma)$ is the log likelihood associated with the $n$ Gumbel observations $\{(y_i, \textbf{x}_i)_{i=1}^n\}$, that is, 
$l(\Gamma) = \sum_{i=1}^n \theta(\textbf{x}_i)-y_i - e^{\theta(\textbf{x}_i)-y_i} $. The integrated square second derivative penalties prevent interpolating estimate functions of $h_k$ and the smoothing parameters $\gamma_k$ control the trade-off between model goodness of fit and model smoothness. When $\gamma_k\to\infty$ the estimate of $h_k(x_k)$ tends to be a straight line.

The minimization of (\ref{GAMLikelihood})  is performed based on an outer-iteration procedure. At each iteration, $\Gamma$ and $\gamma$ are estimated separately by penalized iteratively re-weighted least squares and a prediction error method from generalized cross validation, respectively (for details, see Chapter 5 in \cite{Wood}). Under appropriate regularity assumptions on the functions $h_k$ and the smoothing parameters, the GAM resulting estimator $\hat{\theta}(\textbf{x}) = h\{\textbf{u}^\top\hat{\pmb{\beta}} + \sum_{k=1}^K \hat{h}_k(t_k)\}$ is consistent and asymptotically normal \citep{Yoshida, Wood}. More GAM-related theory is available, for example, in \cite{Wood2}.

\subsection{LogMax-projection in SCM}

Let us move to the causal setup. We assume that we observe variables from different environments $\mathcal{E}$, with the number of environments $|\mathcal{E}|$ being finite. In environment $e\in\mathcal{E}$, let $\textbf{X}^e=(X_1^e, \dots, X_d^e)^\top$ be our covariates, $\textbf{Z}^e=(Z_1^e, Z_2^e)$ be a simple max-stable random vector with covariate-dependent extremal coefficient $\theta(\textbf{X}^e)$. Let $Y^e$ be the corresponding LogMax-projection in environment $e\in\mathcal{E}$. In each environment, we observe $n_e\in\mathbb{N}$ iid observations $(\textbf{X}^e_i, \textbf{Z}^e_i)_{i=1}^{n_e}$.
For $S\subseteq \{1, \dots, d\}$, we denote $\theta(\textbf{X}^e_{S})$ the extremal coefficient between $Z_1, Z_2$ conditioned on $\textbf{X}^e_S=\textbf{x}^e_S$ (as a reminder, we denote $\textbf{X}_S=(X_{i_1}, \dots, X_{i_{|S|}})$ for a set $S=\{i_1, \dots, i_{|S|}\}\subseteq \{1, \dots, d\}$). 

Combining Assumption \ref{Assumption1} and Lemma \ref{LLMax-projection}, we assume that the structural equation is of the form 
$$
Y^e=\log\{\theta(\textbf{X}^e_{S^\star})\} + \varepsilon^e, \varepsilon^e \independent  \textbf{X}^e_{S^\star}, \,\, \varepsilon^e\sim Gumbel(-\gamma, 1),
$$
where $S^\star$ is the true set of parents. The following proposition states that the residuals are asymptotically invariant. 
\begin{proposition}
\label{Theorem 2}
Suppose Assumption \ref{Assumption1} holds. Let $\hat{\theta}$ be an estimator of $\theta$ that is consistent in the sense that  $\hat{\theta}(\textbf{X}_{S^\star}^e) - \theta(\textbf{X}_{S^\star}^e)\overset{P}{\to} 0$ for all $e\in\E$. Define the residuals $\hat{\eta}_{S^\star}^e := Y^e - \log\{\hat{\theta}(\textbf{X}_{S^\star}^e)\}$. Take $f\neq e\in\mathcal{E}$. Then, as the number of data (in both environments) tends to infinity, we have $\hat{\eta}^e_{S^\star}\overset{D}{=} \hat{\eta}^f_{S^\star}$. 
\end{proposition}

The proof is based on the fact that $Y^e-\log\{\theta(\textbf{X}_{S^\star}^e)\}\overset{D}{=} Y^f-\log\{\theta(\textbf{X}_{S^\star}^f)\}$ from Assumption \ref{Assumption1} and that the estimator is consistent. Hence, the invariant residual distribution test (\ref{Invariant residual distribution test}) can be directly used. Note that Proposition \ref{Theorem 2} is formulated only for the set of true direct causes $S^\star$. 

We can directly test the invariance property $\hat{\eta}^e_{S}\overset{D}{=} \hat{\eta}^f_{S}$ for all $S\subseteq\{1,\dots, d\}$. This invariance should hold for the true set $S^\star$ and hence, the same argument as in Theorem \ref{Theorem_Level} can be used. In the following, we propose a methodology for testing this invariance property. 

\subsection{Generic algorithm}
\label{Generic_algorithm}

The proposed methodology can be summarized in the following algorithm.

\begin{enumerate}
\item For each margin separately: 
\begin{itemize}
\item Transform original data to unit Fréchet scale. This can be done by using the block maxima approach to fit a GEV distribution to the univariate series. 
\item The GEV parameters can also depend on other covariates.
\end{itemize}

\item In each environment $e\in\E$ separately:
\begin{itemize}
\item Compute the LogMax-projection $Y^e$, 
\end{itemize}
 
\item For each $S\subseteq\{1,\dots, d\}$:
\begin{itemize}
\item Pool the data from all environments and calculate $\hat{\theta}(\mathbf{X}_S)$,
\item Compute the residuals $\hat{\eta}_{S}^e = Y^e-\log\{\hat{\theta}(\mathbf{X}^e_{S})\}$,
\item Use a $k$-sample test, to test if the residuals from different environments $\hat{\eta}_{S}^{e}$ are equally distributed.
\end{itemize}
\item Get $\hat{S}^\star$ as the intersection of the subsets $S\subseteq \{1, \dots, d\}$ where the $k$-sample test was not rejected. 
\end{enumerate}

In practice, we do not have to repeat the algorithm for all $S\subseteq\{1,\dots, d\}$; if we do not reject the $k$-sample test for a set $S$, testing for a superset of $S$, is redundant.  

Several approaches can be used for transforming the data to unit Fréchet margins. In our application, we fit a covariate-dependent GEV distribution to block maxima. Alternatively, a peak over threshold approach \citep{Pickands1975} can be used as well. 

In our application, we use GAM to estimate  $\theta(\textbf{X}_S)$ based on minimizing (\ref{GAMLikelihood}). This choice was also used in \cite{Linda} and is a default method in non-linear ICP methodology in \cite{Christina}. 

As for the $k$-sample test, we use the Anderson-Darling $k$-sample test (AD-test). In dimension two, the distance between the two samples is $$AD_{n,m}=\frac{mn}{n+m}\int_{-\infty}^{\infty}\frac{(F_n(x)-G_m(x))^2}{H_{n+m}(x)(1-H_{n+m}(x))} dH_{n+m}(x),$$ where $H_{n+m}$ is an empirical distribution function of the pooled datasets. The $k$-sample AD-test is the natural generalization of its two-sample counterpart  \citep{AD-test}.   
It has typically better power than the Kolmogorov-Smirnov test \citep{AD-KStest}. This test is used at a level $\alpha=0.05$ level. Note that an estimation error can arise from two sources: from the estimation of the margins, and from the estimation of $\hat{\theta}$. Although the equality $\hat{\eta}_{S^\star}^{e}\overset{D}{=}\hat{\eta}_{S^\star}^{f}$ should asymptotically hold, we do not claim that a proper $\alpha$ level holds for a finite number of observations. In the following, we illustrate the methodology on a series of simulations.


\section{Simulations}\label{SectionSimulations}

We proceed with a simulation study. Section \ref{4.1} describes the study and results based on a specific example, discussing power with different causal strength, intervention strength and multicollinearity. Section \ref{Simulation_study_hidden_cofounder2} discusses a hidden confounder case, when  Assumption \ref{Assumption1} is not fulfilled. Section \ref{Simulation_study_effect_on_the_margins} deals with a case when we start with non-extreme data; we compare the choices of different thresholds and different marginal transformations and we discuss a case when we do not ``clean the margins” properly. 

We want to stress that we are only interested in inferring causal relations, not an inference on the extremal dependence structure itself.
\subsection{Intervention, causal strength, and multicollinearity}
\label{4.1}

Let $\textbf{Z}=(Z_1, Z_2)^\top$ have a simple bivariate max-stable distribution and $\mathbf{X}=(X_1,X_2,X_3)^\top$ be a set of covariates. We assume the distribution of $\mathbf{Z}$ to be parameterized by a logistic copula \eqref{LogisticModel} with parameter $\alpha=\alpha(X_1)=\frac{e^{\beta X_1}}{1+e^{\beta X_1}}$. Since $\alpha$ is a function of $X_1$ only, the true set of direct causes of our extremal dependence is therefore $S^\star = \{1\}$. We consider $X_2, X_3$ as other covariates, graphically represented by a SCM in the form $X_2\leftarrow X_1\to Y\to X_3$, where $Y=\log\{\max(Z_1, Z_2)\}-\gamma$ is the LogMax-projection. We observe the covariates $X_1, X_2, X_3$ and the change of environments are interventions on $X_1, X_2$, and $X_3$. Details about the generation of the random variables can be found in Table \ref{Table_for_simulation1} and the three parameters $p,q$, and $\beta$ are as follows
\begin{itemize}
\item the parameter $p$ describes the strength of intervention. If $p=0$ then there is no intervention on $X_1$,
\item the parameter $\beta$ represents the causal strength of $X_1$ on $Y$. If $\beta=0$ then $Y$ is independent of $X_1$,
\item the parameter $q$ describes the dependence between $X_1$ and  $X_2$. If $q=0$ then $cor(X_1, X_2)=0$. Moreover, if $q=0.25$, the correlation $cor(X_1, X_2)\approx 0.25$ and if $q=1$ then $cor(X_1, X_2)\approx 0.75$. 
\end{itemize}

We proceed with our causal methodology using the algorithm described in Section \ref{Generic_algorithm}. We use $50$ repetitions for different choices of $q, p,\beta$ and for two sample sizes $n=500$ and $n=1500$.  Figure~\ref{Figure789} describes how often we get $\hat{S}^\star=S^\star = \{1\}$ using our methodology. The upper axis in Figure~\ref{Figure789} represents the expectation $\mathbb{E}(\alpha)$ in the second environment under the corresponding intervention strength $p$ (note that  $\mathbb{E}(\alpha)=1/2$ in the first environment).  Obviously, the larger the intervention $p$, the higher the heterogeneity between the environments and hence the easier it is to find the correct causal set. The reason why we obtain worse results with a larger absolute value $q$ is that it becomes statistically harder to distinguish between correlated $X_1, X_2$ since both are good predictors (with an exception in the third environment). Not surprisingly, the larger the causal strength $\beta$, the better the results we obtain. 

\begin{table}[]
\caption{Definitions of environments for the simulation study in Section~\ref{4.1} with graphical representation $X_2\leftarrow X_1\to Y\to X_3$. The second and the third environments consist of interventions on $X_1, X_2$. The fourth environment changes $X_3$ by adding an additional relation $X_1\to X_3$. Note that Assumption \ref{Assumption1} is valid since $Y\mid X_1$ is the same in all environments. }
\label{Table_for_simulation1}
\begin{tabular}{|l|l|}
\hline
\multicolumn{1}{|c|}{\cellcolor[HTML]{EFEFEF}\textbf{First environment}}                 & \multicolumn{1}{c|}{\textbf{Second environment}}                                                     \\ \hline
\cellcolor[HTML]{EFEFEF}$X_1=\varepsilon_1$, where $\varepsilon_1\sim N(0,1)$            & $X_1=\textcolor{red}{p}+ \varepsilon_1$, where $\varepsilon_1\sim N(0,1)$                            \\ \hline
\cellcolor[HTML]{EFEFEF}$X_2= qX_1 + \nu,$ where $\nu\sim Exp(1)$                        & $X_2= qX_1 + \nu,$ where $\nu\sim Exp(1)$                                                            \\ \hline
\cellcolor[HTML]{EFEFEF}$\alpha=\alpha(X_1)=\frac{e^{\beta X_1}}{1+e^{\beta X_1}}$       & $\alpha=\alpha(X_1)=\frac{e^{\beta X_1}}{1+e^{\beta X_1}}$                                           \\ \hline
\cellcolor[HTML]{EFEFEF}$X_3 = Y + \varepsilon_3$, where $\varepsilon_3\sim N(0,1)$ & $X_3 = Y + \varepsilon_3$, where $\varepsilon_3\sim N(0,1)$                                     \\ \hline
\multicolumn{1}{|c|}{\textbf{Third environment}}                                         & \multicolumn{1}{c|}{\cellcolor[HTML]{EFEFEF}\textbf{Fourth environment}}                             \\ \hline
$X_1=\varepsilon_1$, where $\varepsilon_1\sim N(0,1)$                                    & \cellcolor[HTML]{EFEFEF}$X_1=\varepsilon_1$, where $\varepsilon_1\sim N(0,1)$                        \\ \hline
$X_2= \textcolor{red}{-10}$                                                              & \cellcolor[HTML]{EFEFEF}$X_2= qX_1 + \nu,$ where $\nu\sim Exp(1)$                                    \\ \hline
$\alpha=\alpha(X_1)=\frac{e^{\beta X_1}}{1+e^{\beta X_1}}$                               & \cellcolor[HTML]{EFEFEF}$\alpha=\alpha(X_1)=\frac{e^{\beta X_1}}{1+e^{\beta X_1}}$                   \\ \hline
$X_3 = Y + \varepsilon_3$, where $\varepsilon_3\sim N(0,1)$                         & \cellcolor[HTML]{EFEFEF}$X_3 = Y+ \textcolor{red}{X_1}+\varepsilon_3$, where $\varepsilon_3\sim N(0,1)$ \\ \hline
\end{tabular}
\end{table}

\begin{figure}[]
\centering
\includegraphics[scale=0.47]{ 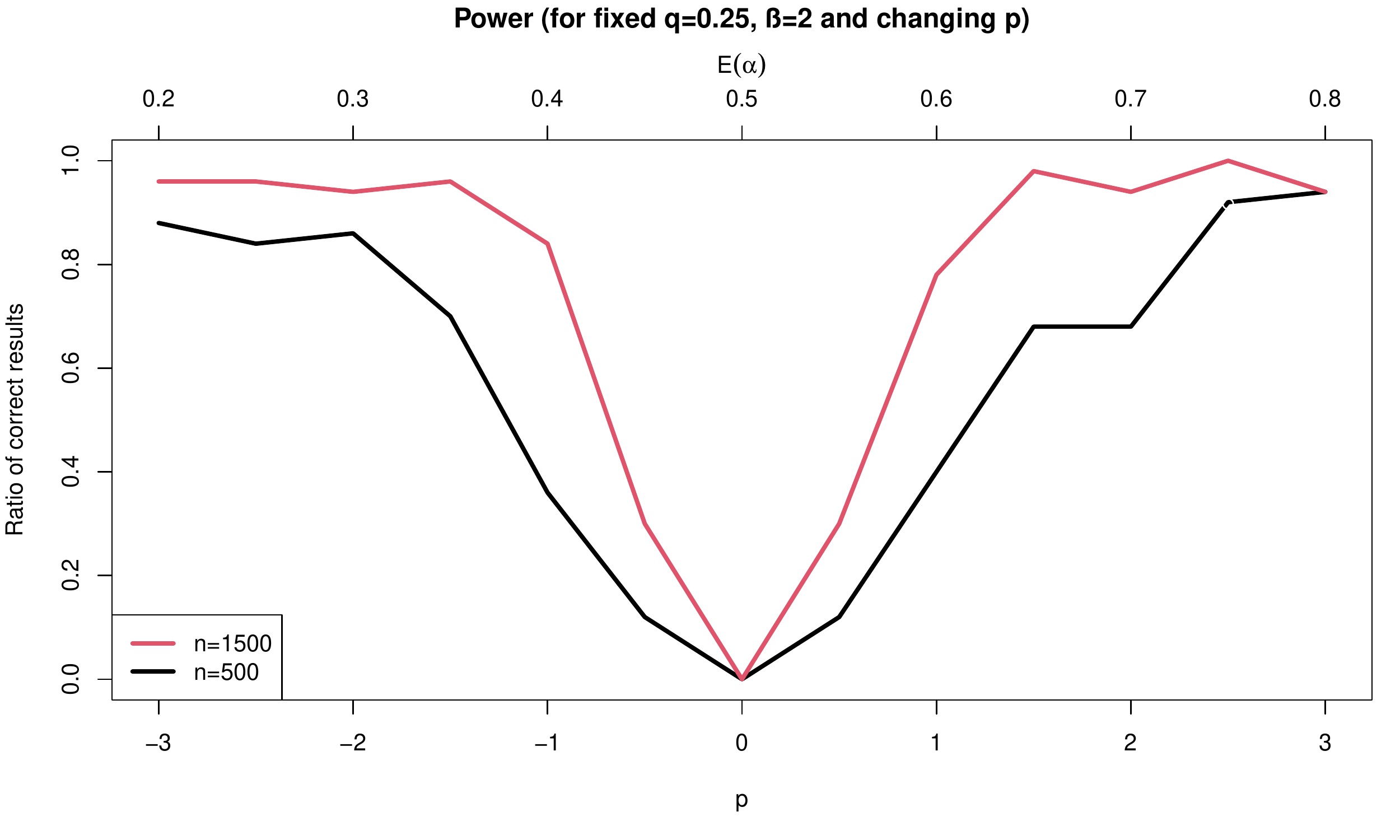}
\includegraphics[scale=0.47]{ 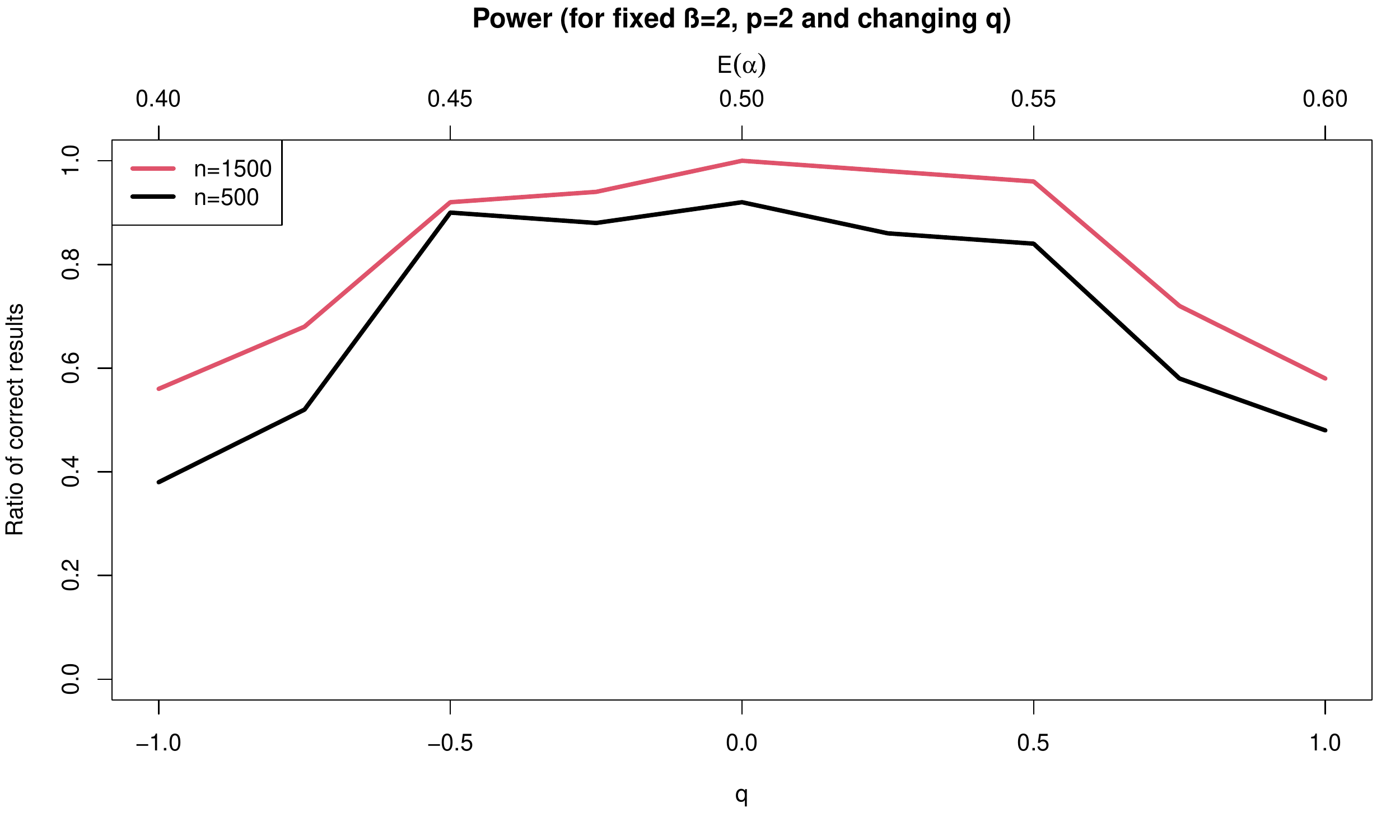}
\includegraphics[scale=0.47]{ 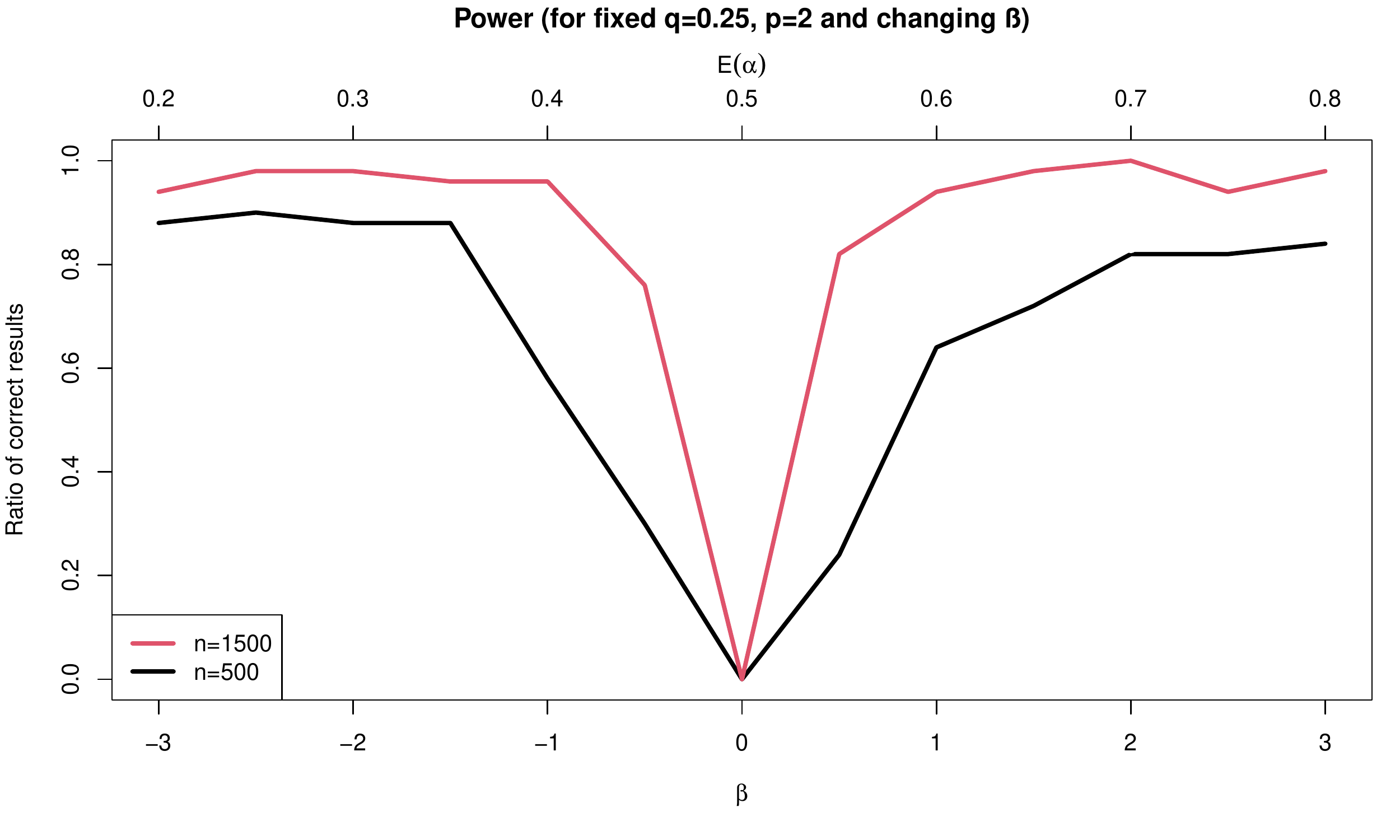}
\caption{Ratio of correct results for the different choices of parameters $p,\beta, q$ as described in Table~\ref{Table_for_simulation1} and for two sample sizes $n=500, 1500$. The upper axis represents the expectation $\mathbb{E}(\alpha)$ in the second environment.}
\label{Figure789}
\end{figure}

\subsection{Hidden confounder }
\label{Simulation_study_hidden_cofounder2}

We present a situation in which a hidden confounder of a causal covariate and the extremal dependence coefficient are present. The graphical representation of the model can be found in Figure \ref{Figure_Graph_HC2}.

\begin{figure}[]
\centering
\includegraphics[scale=0.35]{ 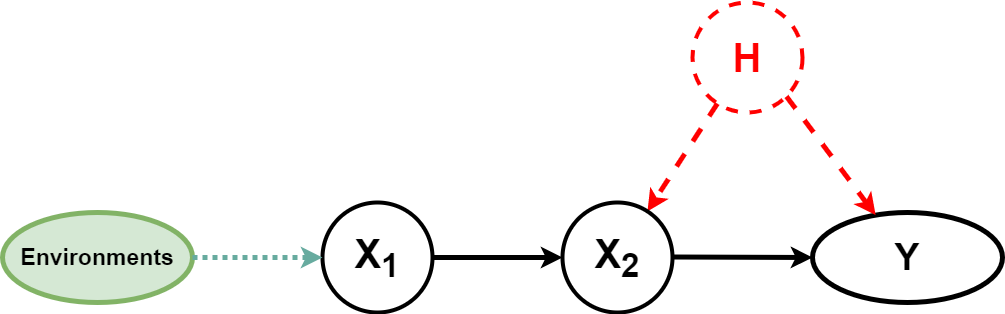}
\caption{Graphical representation of causal relations in the simulation study with a hidden confounder $H$. The green arrow from the environment represents the interventions described in subsection \ref{Simulation_study_hidden_cofounder2}.     }
\label{Figure_Graph_HC2}
\end{figure}

 Let $\textbf{Z}=(Z_1, Z_2)^\top$ have a simple bivariate max-stable distribution, $X_1, X_2$ are the observed covariates and $H$ is an unobserved covariate. We assume the distribution of $\mathbf{Z}$ to be parameterized by a Hüsler–Reiss copula (\ref{HR}) with parameter $\lambda=\lambda(X_2, H)=|X_2 + 3H|/5$. 

We intervene on $X_1$ and generate the variables in the corresponding environments:
\begin{enumerate}
\item $H+1\sim Exp(1)$, $X_1= \varepsilon_1, X_2 = X_1 + H + \varepsilon_2$ where $\varepsilon_1, \varepsilon_2\overset{iid}{\sim} N(0,1)$,
\item $H+1\sim Exp(1)$, $X_1= \varepsilon_1 + \color{red}{5}$, $X_2 = X_1 + H + \varepsilon_2$ where $\varepsilon_1, \varepsilon_2\overset{iid}{\sim} N(0,1)$.
\end{enumerate}
Moreover, in both environments, we compute the LogMax-projection~$Y$. 

Finally, we compute $\hat{S}^\star$ based on $X_1, X_2,Y$ only. Note that the only observed direct cause of $Y$ is $X_2$. However, due to the hidden variable $H$, Assumption \ref{Assumption1} no longer holds. It is easy to see that $Y\mid X_1$ is invariant through environments, however, $Y\mid X_2$ is not due to the hidden confounder. Simulations show how much this heuristic is justified in practice. 

For a range of sample sizes $n$, we generate the variables as described above, compute $\hat{S}^\star$ and repeat $50$ times. The results are described in Figure \ref{Figure_HiddenConfounder2}. We can see that for smaller $n$, we tend to \textit{not} reject $H_{0,S}$ for both $S=\{2\}, S=\{1\}$. Hence, the resulting estimate is an empty set (intersection of $\{1\}\cap\{2\}$). 
For a larger $n$, the test can recognise the effect of the hidden confounder and we tend to reject $H_{0, S}$ for $S=\{2\}$, and hence, the resulting estimate is $\hat{S^\star} = \{1\}$. 

Note that if the effect of $H$ is small, we tend to accept $H_{0,S}$ for $S=\{2\}$ and also for $S=\{1\}$ which results in $\hat{S^\star} = \emptyset$. On the other hand, if the effect of $H$ is large (larger than the effect of intervention), then the effect of $X_1$ on $Y$ will be (proportionally) small and we tend to accept $H_{0, S}$ for $S=\emptyset$. The statistical guarantees such as Theorem \ref{Theorem_Level} or Theorem \ref{Theorem 2} for  $\hat{S}^\star$ are no longer valid. However, following \citep{Peters}[Proposition 5],  $\hat{S}^\star$ will estimate some subset of its ancestors (instead of parents). For a more detailed discussion of hidden confounders in the non-extremal case, see  \citep{Peters}[Appendix C].

\begin{figure}[]
\centering
\includegraphics[scale=0.45]{ 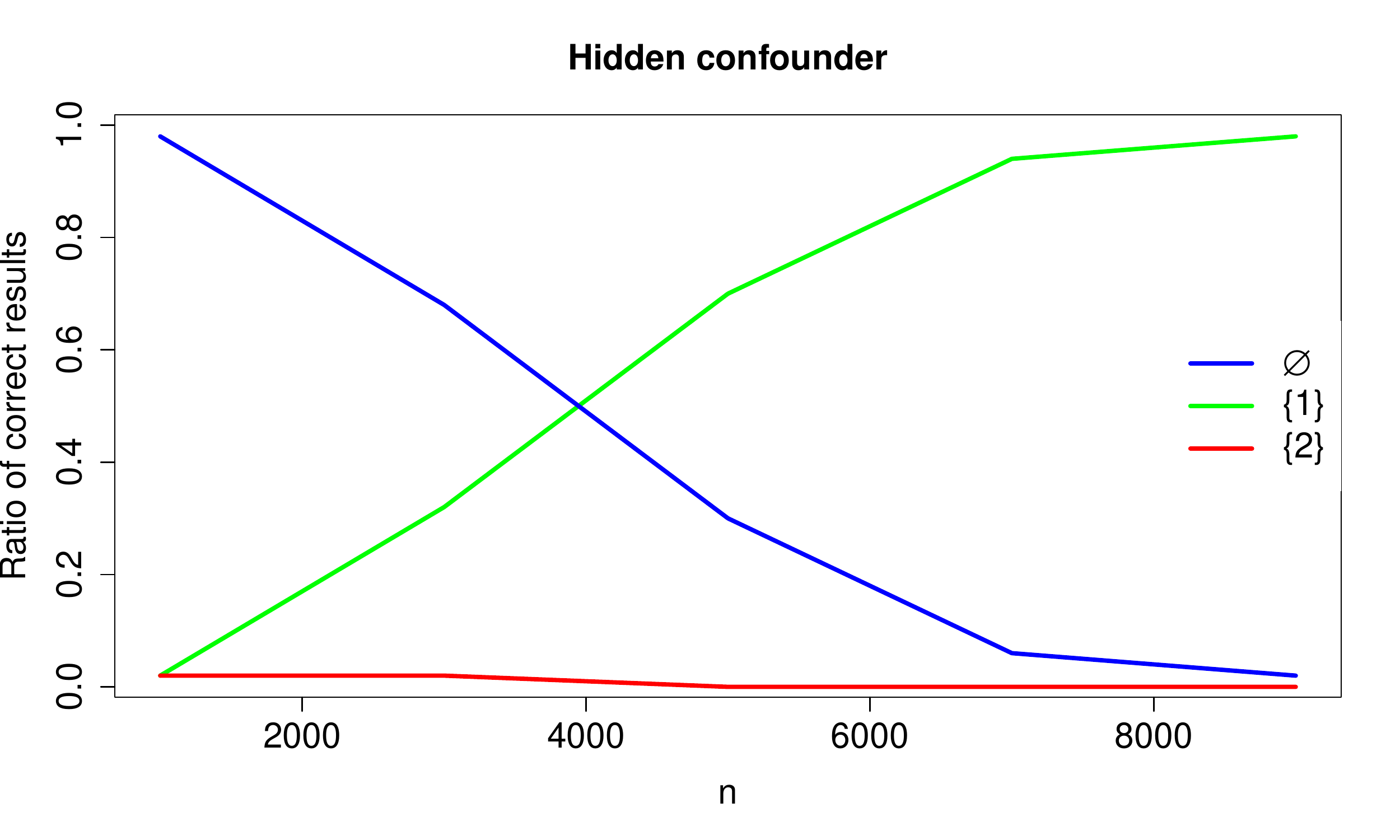}
\caption{ The figure corresponds to the simulation study from Section \ref{Simulation_study_hidden_cofounder2}, with a hidden confounder $H$. The figure represents the final estimate $\hat{S}^\star$. For a smaller $n$, the estimate $\hat{S}^\star = \emptyset$, and for a larger $n$,  the estimate $\hat{S}^\star = \{1\}$. 
}
\label{Figure_HiddenConfounder2}
\end{figure}

\subsection{Effect on the margins }
\label{Simulation_study_effect_on_the_margins}

In this section, we consider the case when the observed covariate has an effect on the margins as well as on the dependence structure. Suppose we observe two covariates $X_1, X_2$, where $X_1$ has an effect on the extremal dependence of $\textbf{Z}=(Z_1, Z_2)$ and $X_2$ has an effect on the margins $Z_1$ and $Z_2$. The graphical representation can be seen in Figure \ref{Figure_Graph_marginals}. 

\begin{figure}[]
\centering
\includegraphics[scale=0.4]{ 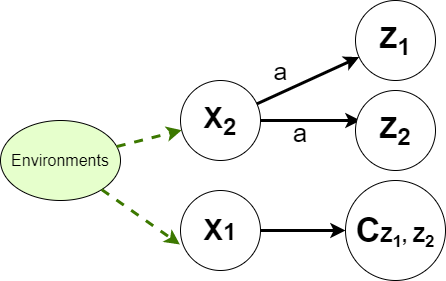}
\caption{Graphical representation of causal relations in the simulation study with different marginals. Here, $C_{Z_1, Z_2}$ represents the dependence structure and $Z_1, Z_2$ are marginals. Green arrows from the environment represent the interventions described in subsection \ref{Simulation_study_effect_on_the_margins}.     }
\label{Figure_Graph_marginals}
\end{figure}

We simulate the following data: for $n, \tau\in\mathbb{N}$, we consider $\tilde{\textbf{Z}}^1, \tilde{\textbf{Z}}^2, \dots, \tilde{\textbf{Z}}^{\tau\cdot n}$ are bivariate iid random variables distributed as a generic element $\tilde{\textbf{Z}}=(\tilde{Z}_1, \tilde{Z}_2)$, where margins $\tilde{Z}_1 = a X_2+\varepsilon_1, \tilde{Z}_2 = a X_2+\varepsilon_2$, where $\varepsilon_1, \varepsilon_2\overset{iid}{\sim}N(0,1)$ for a parameter $a\geq 0$. The dependence structure between $\tilde{Z}_1, \tilde{Z}_2$ is modeled by a  logistic copula with parameter $\alpha=\alpha(X_1)$. Hence, the dependence structure (as well as the extremal dependence structure) is a function of $X_1$ only.

We transform $\tilde{\textbf{Z}}^1, \tilde{\textbf{Z}}^2, \dots, \tilde{\textbf{Z}}^{\tau\cdot n}$ into a $\tau$ period-maxima, that is, we define $$\textbf{Z}^1:=\max(\tilde{\textbf{Z}}^1, \tilde{\textbf{Z}}^2, \dots, \tilde{\textbf{Z}}^{\tau}), \textbf{Z}^2:=\max(\tilde{\textbf{Z}}^{\tau+1}, \tilde{\textbf{Z}}^{\tau+2}, \dots, \tilde{\textbf{Z}}^{2\tau}), \dots,$$ where the maxima is component-wise maxima. 

Hence, we obtain independent maxima $\textbf{Z}^i=(Z^i_1,Z^i_2), i=1, \dots, n$. We model their margins by a GEV distribution with parameters depending on $X_2$. By fitting the covariate-dependent GEV distribution, we transform the margins into unit Frechet using (\ref{FrechMarg}). We model GEV parameters $$\mu=\mu_0+f_1(X_2), \sigma=\sigma_0 + f_2(X_2), \xi=\xi_0$$ and estimate them smoothly using the GAM framework. Finally, we compute the LogMax-projections $Y^1, \dots, Y^{{n}}$. 

We intervene on the covariates and observe the system. More precisely, we generate the variables in the corresponding environments as follows:
\begin{enumerate}
\item $X_1= 1$, $X_2 =\varepsilon_2$ where $\varepsilon_2\overset{iid}{\sim} Exp(1)$,
\item $X_1= \color{red}{5}$, $X_2 =\varepsilon_2$ where $\varepsilon_2\overset{iid}{\sim} Exp(1)$,
\item $X_1= 1$, $X_2 \overset{}{=}\color{red}{1}$.
\end{enumerate}
We consider the copula parameter in the form $\alpha(X_1) = \frac{1}{X_1}$. In both environments, we generate $\tilde{\textbf{Z}}^1, \tilde{\textbf{Z}}^2, \dots, \tilde{\textbf{Z}}^{\tau\cdot n}$ and compute the LogMax-projections as mentioned before. 

Finally, we estimate the set of direct causes of the extremal dependence structure following points 3 and 4 in Section \ref{Generic_algorithm}.  We repeat $50$ times for a range of parameters $a\in [0,2]$, $n=1000$ and $\tau\in\{50, 100\}$, where $a$ represents the effect of $X_2$ on the margins, $n$ is the size of the final dataset (of extremes) and $\tau$ is a size of the block maxima we take. 
The blue lines in Figure \ref{Figure_marginal_effect} represent the ratio of correctly estimating $\hat{S}^\star = \{X_1\}$. 

Now, consider a different question: What happens if we do \textit{not} include $X_2$ for the marginal modelling? 
We proceed with the same setup as before, but we model GEV parameters $\mu=\mu_0, \sigma=\sigma_0, \xi=\xi_0$ and estimate them using maximum likelihood.

Again, we estimate the set of direct causes of the extremal dependence structure following points 3 and 4 in Section \ref{Generic_algorithm}.  We repeat $50$ times for a range of parameters $a\in [0,2]$, $n=1000$ and $\tau\in\{50, 100\}$. The red lines in Figure \ref{Figure_marginal_effect} represent the ratio of correctly estimating $\hat{S}^\star = \{X_1\}$. 

We can see that ignoring the marginal effects of $X_2$ creates a spurious causal effect on the extremal dependence structure. Indeed, if $X_2$ is large, the margins of $\textbf{Z}$ tend to be both very large. Hence, extremes occur concurrently even though $X_2$ has no effect on the extremal copula, and our results are compromised. 

The same holds if $X_2$ is hidden; if there exists a hidden covariate that affects both margins with ``enough” strength, it can affect the inference about the extremal dependence structure. Therefore, it is good practice to include as many relevant covariates for the marginal inference as possible.

\begin{figure}[]
\centering
\includegraphics[scale=0.42]{ 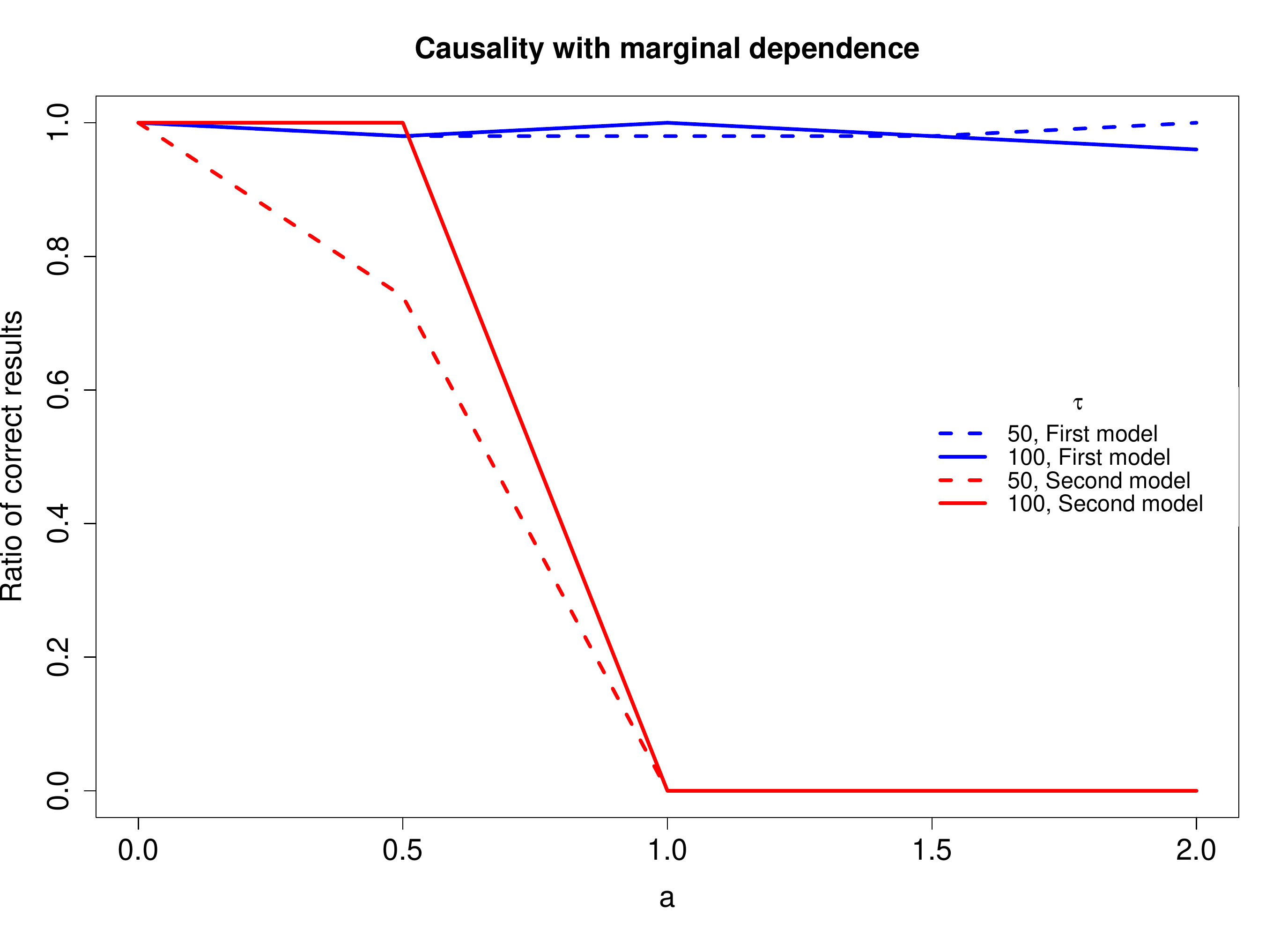}
\caption{ Simulation study from Section \ref{Simulation_study_effect_on_the_margins}, where $X_2$ has an effect on the marginals with strength $a\geq 0$ and the marginal block-maxima are of size $\tau$. The blue line correspond to the case in which we include $X_2$ in marginal modelling. The red line correspond to the case in which we do not include $X_2$ into marginal modelling. The lines represent the ratio of results when $\hat{S}^\star=\{1\}$. Each time when  $\hat{S}^\star\neq \{1\}$, then $\hat{S}^\star=\emptyset$. 
}
\label{Figure_marginal_effect}
\end{figure}


\section{NO$_2$ Application}

We illustrate the causal methodology for the tail dependence on a dataset of $\text{NO}_2$ concentration levels measured in the UK  \citep{openair}. To explain the extremal coefficient $\theta$ between the $\text{NO}_2$ levels of two different stations, we consider three different potential causal covariates $\textbf{X}=(X_1, X_2, X_3)^\top$. Namely, the distance (in kilometres) between the stations, the time (year) of the observation, and the type of area (traffic/non-traffic). A similar dataset is studied in  \cite{LindaNO2}. We wonder which of these covariates have a causal effect on the tail dependence. 
The choice of the time as a covariate may sound unusual as we expect that the physical mechanism is inherently given by laws of nature and does not change over time. However, it might be considered as a proxy for potential confounders such as the increase in temperature or $\text{CO}_2$.

\subsection{Transforming margins and LogMax-projecting}
We follow steps 1 and 2 from the algorithm (\ref{Generic_algorithm}). We transform our observations into bivariate simple max-stable ones. The first step is therefore transforming the margins to the unit Fréchet scale using transformation (\ref{FrechMarg}). 
We use weekly maxima (around 900 weeks at each station; recall that we have hourly measurements) and model these weekly maxima by a GEV distribution  \eqref{GEV} with time-varying parameters 
\begin{align*}
\mu(t_1, t_2)&=\mu_0 + f_1(t_1)+f_2(t_2),   \\
\sigma(t_1, t_2)&=\sigma_0 + g_1(t_1)+f_2(t_2),
\end{align*}
where $t_1$ represents the year and $t_2$ represents the week of the year, and $f_1, f_2, g_1, g_2$ are smooth functions. Inference is performed under the GAM framework \citep{evgam} and likelihood ratio tests are conducted to assess whether we need smoothly varying terms or whether parametric and sinusoidal terms already provide a good fit. The final fitted model is then used to transform the data at each station to the unit Fréchet distribution using the probability integral transform (\ref{FrechMarg}). Figure \ref{QQplot} shows a QQ-plot of the resulting transformed observations for one of the stations (RBKC Chelsea in London) against the unit Fréchet. It shows a reasonable fit besides a few outliers. 

The next step of the algorithm (\ref{Generic_algorithm}) is to pair the stations and compute the LogMax-projection. We pair only the stations with the same traffic type. We end up with 106 pairs of stations. For each pair, we compute the LogMax-projection transformation described in Lemma \ref{LLMax-projection}. 

\begin{figure}[]
\centering
\includegraphics[scale=0.9]{ 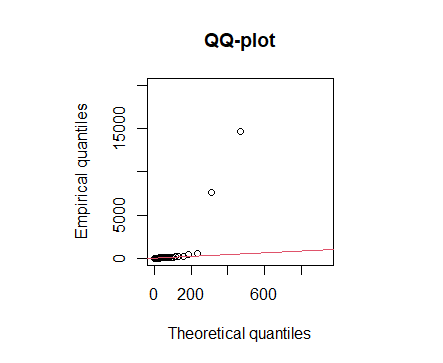}
\caption{QQ-plot for the transformed values against the unit Fréchet distribution corresponding to the station Newham Wren Close in London.}
\label{QQplot}
\end{figure}


\subsection{Defining environments}
\label{DefiningEnvironments}
In this application, we define one environment as one pair of stations.
Consider two pairs of stations for which all covariates remain unchanged, but $X_3$ changes from non-traffic to traffic. This can be considered as a "do-intervention" on the covariate $X_3$ since we assume that the underlying physical mechanism is the same regardless of the location of our stations. A few remarks on the choice of such environments are in order. 

\begin{remark}
Notice that this is not simple conditioning on traffic/non-traffic stations. Take measurements from one pair of stations. Imagine that we conduct an experiment where we artificially add traffic to those stations to see if the extremal coefficient changes. Instead of adding traffic to one station, we measure a second (different) pair that has the same value of covariates, except for traffic. The main difference here is that we assume that the data-generating process remains unchanged, apart from the traffic covariate. 
\end{remark}

\begin{remark}[Possible hidden variable]
The main drawback is the following: Is there a hidden variable that differs between the pairs and causally affects our response? This is represented in Figure \ref{Figure88}. We need to assume that there is no such variable or if there is, it has a negligible impact on the response variable.

\begin{figure}[]
\centering
\includegraphics[scale=0.4]{ 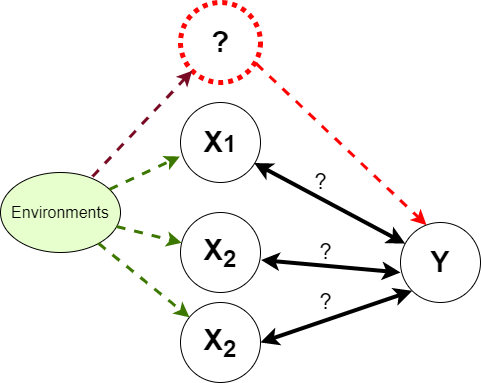}
\caption{Possible hidden variable in the application. }
\label{Figure88}
\end{figure}
\end{remark}

\begin{remark}[Time]
Each station is measured over a different period. Some stations started in 2003, some in 1997. For each pair, we take only the intersection of measured weeks. Hence, we can argue that by changing the pair, we also change the period of observations. However, this change is minor and almost negligible. Our findings support the fact that time is not a causal covariate. This might of course be due to changes between environments not being strong enough to find such evidence. 
\end{remark}

\subsection{Details for the choices of environments and results for one choice of stations}

We take five different environments, that is, five pairs of stations (10 different stations in total). We want the interventions to be ``strong”. For $X_1$, we choose one pair from each of the bins of distances described in Table \ref{bins_table}. As for $X_3$, we choose three pairs with the type traffic and two with the type non-traffic. 
In order to make an intervention on $X_2$, we would need to divide the dataset into different periods and compare them. However, due to the limited time span of our dataset, this is not an option, so we cannot intervene on $X_2$.

\begin{table}[]
\centering
\begin{tabular}{ |c|c|c|c|c| } 
 \hline
 <10km & (10km, 20km) & (20km, 50km) &(50km, 100km)&100km<  \\ 
 \hline
\end{tabular}
\caption{Definition of bins of distances. We aim to have one pair of stations for each bin.}
\label{bins_table}
\end{table}

We explain now the results on one specific choice for the environments. This choice is described in Table \ref{myTable}. In Section \ref{Repeating_application}, we repeat the experiments with 50 different choices to get more robust results.  We estimate $S^\star$ using the methodology described in Section \ref{Generic_algorithm}. For each $S\subseteq \{1,2,3\}$ we test the invariance property using step 3 in algorithm \ref{Generic_algorithm}. Results (for the specific choice of environments described in Table \ref{myTable}) are described in Table \ref{Table_results}. The hypotheses corresponding to sets $\{1\}, \{1,2\}, \{1,3\}, \{1,2,3\}$ were not rejected, and hence, $\hat{S}^\star= \{1\}\cap \{1,2\}\cap \{1,3\}\cap \{1,2,3\} = \{1\}$. 

\begin{table}[]
\begin{tabular}{@{}llll@{}}
\toprule
Pairs of sites                                    & Distances (km) & Range of years & Type       \\ \midrule
Knightsbridge n.3- Knightsbridge n.4      & 1.37           & 2000-2018      & Traffic    \\
Slough Colnbrook n.3-Slough Colnbrook n.4 & 11.36          & 2000-2018      & Non-traffic \\
Watford Town Hall - Newham                & 46.61          & 2008-2018      & Traffic    \\
Hillingdon South Ruislip - Cambridge n.4  & 94.90          & 2002-2017      & Traffic    \\
Trafford Moss Park - Hounslow Chiswick    & 316.50         & 1999-2018      & Non-traffic \\ \bottomrule
\end{tabular}
\caption{Pairs defining the five environments. }
\label{myTable}
\end{table}

\begin{table}[]
\centering
\begin{tabular}{@{}|l|l|@{}}
\toprule
\rowcolor[HTML]{FFFFFF} 
Choice for a set S & P-values                      \\ \midrule
\rowcolor[HTML]{FFFFFF} 
Empty set          & \cellcolor[HTML]{FD6864}1e-07 \\ \midrule
\rowcolor[HTML]{FFFFFF} 
1                  & \cellcolor[HTML]{38FFF8}0.11  \\ \midrule
\rowcolor[HTML]{FFFFFF} 
2                  & \cellcolor[HTML]{FD6864}0     \\ \midrule
\rowcolor[HTML]{FFFFFF} 
1,2                & \cellcolor[HTML]{38FFF8}0.19  \\ \midrule
\rowcolor[HTML]{FFFFFF} 
3                  & \cellcolor[HTML]{FD6864}6e-07 \\ \midrule
\rowcolor[HTML]{FFFFFF} 
1,3                & \cellcolor[HTML]{38FFF8}0.30  \\ \midrule
\rowcolor[HTML]{FFFFFF} 
2,3                & \cellcolor[HTML]{FD6864}0     \\ \midrule
\rowcolor[HTML]{FFFFFF} 
1,2,3              & \cellcolor[HTML]{38FFF8}0.41  \\ \midrule
\rowcolor[HTML]{FFFFFF} 
Intersection       & \textit{\textbf{1}}           \\ \bottomrule
\end{tabular}
\caption{Resulting p-values for one specific choice of the stations from Table \ref{myTable}. The p-values are computed for each subset of our covariates. It seems that $X_1$ is the only identified causal covariate.  }
\label{Table_results}
\end{table}

\subsection{Repeating results with different stations and discussing results}
\label{Repeating_application}

We repeat the same steps with different choices of stations in order to obtain some confidence bounds for each p-value. In each repetition, we randomly choose five pairs of stations satisfying the following: 

\begin{itemize}
\item There is one pair for each bin of distances from Table \ref{bins_table},
\item Each station is included at most once,
\item There are three pairs from traffic type and two from non-traffic type.
\end{itemize}

Then, we compute the corresponding p-values for each subset $S$. We use 50 repetitions, and boxplots for the p-values associated with each subset of indices are drawn in Figure \ref{boxplot}. Again, this analysis suggests that $X_1$ is the only causal covariate.

\begin{figure}[]
\centering
\includegraphics[scale=0.7]{ 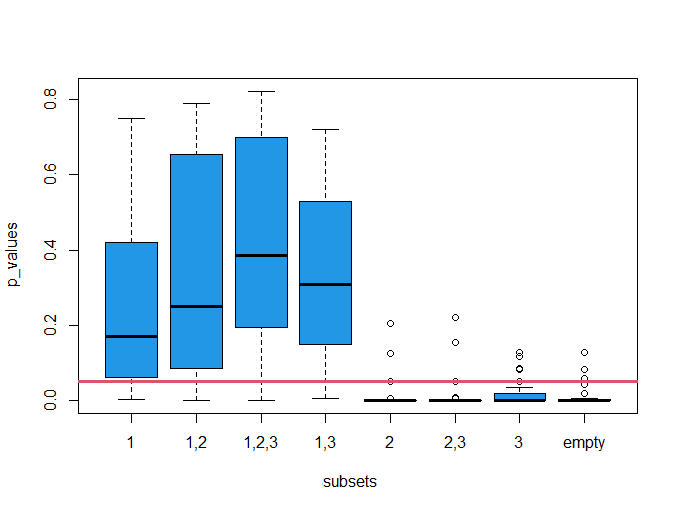}
\caption{ Boxplots of p-values for each subset of $\{1,2,3\}$ when randomly choosing 10 stations 50 times (fulfilling the constraints defined in Section \ref{Repeating_application}). }
\label{boxplot}
\end{figure}

This result is not surprising. It is expected that the extremal dependence decreases with increasing distance. The effect of $X_1$ on the extremal dependence is shown in Figure \ref{effect_of_distance}. 
The fact that $X_2$ was not marked as a cause is not surprising either. We would need much stronger interventions and changes in time observations to detect anything. Additionally, it is unclear how time can cause extremal dependence (there may be several time-changing factors, but we do not know them explicitly). 
As for $X_3$, it seems that the type of the site is not a direct cause of the extremal dependence. 


\begin{figure}[h]
\centering
\includegraphics[scale=0.2]{ 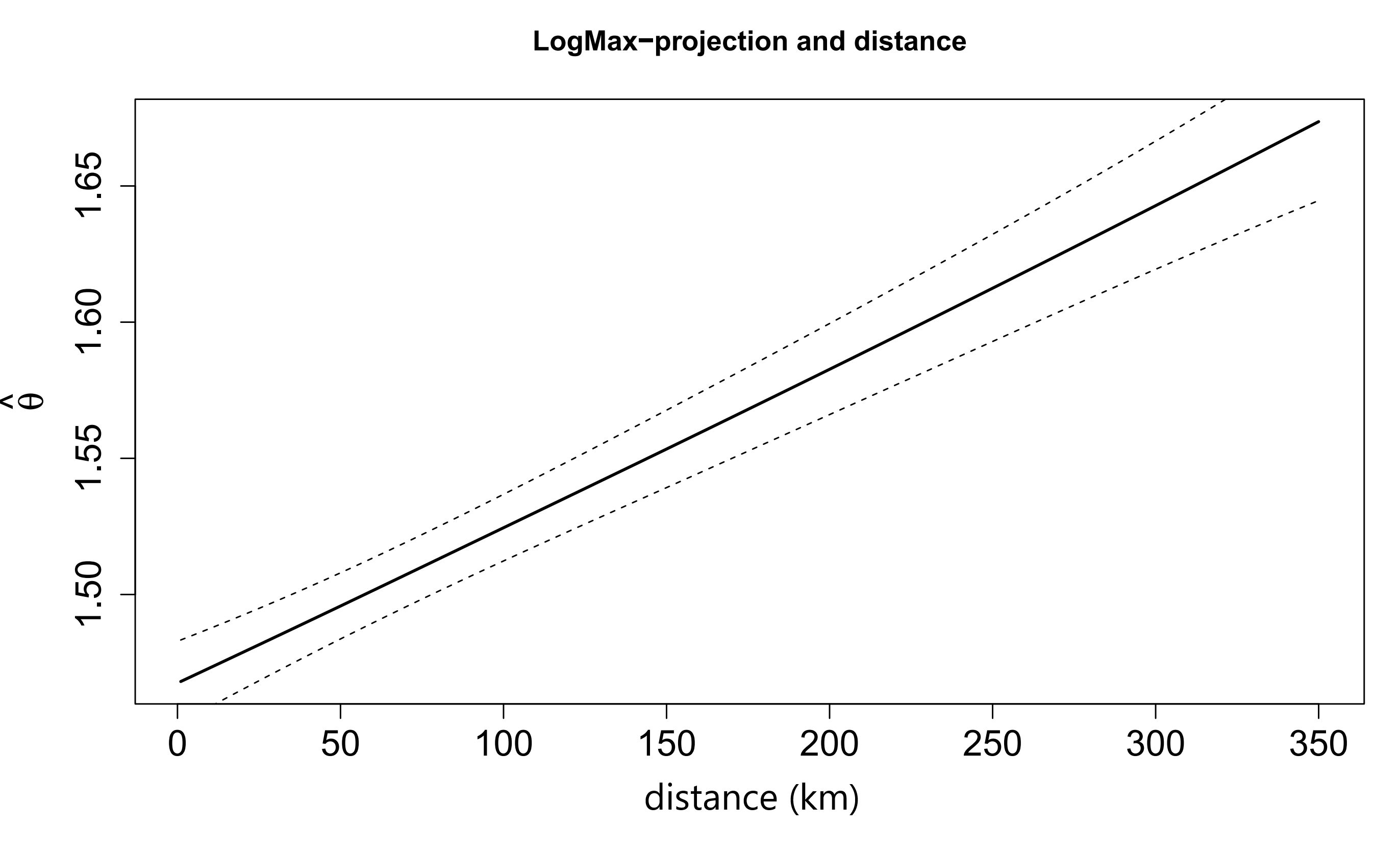}
\caption{ Smooth estimation of the effect of distance on the LogMax-projection. It is based on all available pairs of stations pooled together (With four exceptions. If there is a triplet of stations closer than 5km from each other, we only take one pair). The estimate may be unreliable since the same station is used in several pairs, which creates a dependence between some pairs. We included only the causal covariate. The estimation and the confidence intervals are computed using the GAM framework by regressing the LogMax-projection on the distance.
We can see that the larger the distance, the larger the extremal coefficient (and smaller the extremal dependence).}
\label{effect_of_distance}
\end{figure}

\section{Conclusion}

This paper addresses the problem of causal discovery in the area of multivariate extremes. For a covariate-dependent extremal coefficient  $\theta(\textbf{X})$ between two variables, we discover which of the covariates are not only significant but have a causal influence on the joint tail behavior. We introduced a LogMax-projection, which serves as a basis variable for inference about the extremal coefficient. A causal methodology from \cite{Peters} was adapted in our context to estimate the causal covariates. Algorithm (\ref{Generic_algorithm}) provides a consistent methodology. We used a very flexible GAM modelling for statistical inference.

We stress that the results rely on the absence of a hidden confounder, an assumption that can and should be debated in many situations. The methodology is robust against hidden covariates to some extent. However, if there is a hidden covariate, it must be invariant between environments. We discussed in simulations in Section \ref{SectionSimulations} the case in which this assumption is violated. 

 We believe that this methodology can be very useful in many fields. More practical examples come from econometrics. We often see an extremal dependence between stocks. If one stock has a massive fall, other stocks seem to follow this fall too. Exploiting the causes of such behavior can be very useful and can broadly impact the field. 

Possible future work can relate to the more general multivariate extensions of our methodology. We introduced LogMax-projection to work only with a univariate variable representing a complex multivariate tail dependence structure. However, it can be interesting to work directly with a tail dependence structure of multivariate random variables with dimensions more than $2$ and infer which covariates are significant and causal.


\section*{Conflict of interest}
The authors declare that they have no known competing financial interests or personal relationships that could have appeared to influence the work reported in this paper.

\section*{Acknowledgements}
The work was supported by the Swiss National Science Foundation.

\clearpage 
\bibliography{bibliography}

\end{document}